\documentclass{llncs}
\usepackage{amsmath}
\usepackage{amssymb}
\usepackage{enumerate}
\usepackage{amsfonts}
\usepackage{graphicx}
\usepackage{xspace}
\usepackage{mathrsfs}
\usepackage{url}
\usepackage{etoolbox}
\usepackage{pgfplots}
\pgfplotsset{/pgf/number format/use comma,compat=newest}
\usepackage{colortbl}

\newcounter{instr}
\newcommand{\ninstr}{\refstepcounter{instr}\theinstr.}

\newcommand{\bigO}{\mathcal{O}}
\newcommand{\best}[1]{\underline{#1}}

\newcommand{\eq}{\approx}
\newcommand{\nr}{\textsc{nr}\xspace}
\newcommand{\no}{\textsc{no}\xspace}
\newcommand{\gte}{\beta\xspace}

\newcommand*{\qeda}{\hfill\ensuremath{\blacksquare}}

\title{Efficient Algorithms for the\\Order Preserving Pattern Matching Problem\thanks{This work has been supported by the Scientific \& Technological Research Council Of Turkey (TUBITAK), the Department Of Science Fellowships \& Grant Programs (BIDEB), 2221 Fellowship Program, and by G.N.C.S., Istituto Nazionale di Alta Matematica ``Francesco Severi''.}}

\author{Simone Faro$^\dag$ \and M. O\u{g}uzhan K\"{u}lekci$^\ddag$}
\institute{
$^\dag$Universit\`a di Catania, Department of Mathematics and Computer Science, Italy\\
$^\ddag$Istanbul Medipol University, Department of Biomedical Engineering, Turkey\\
\email{faro@dmi.unict.it, okulekci@medipol.edu.tr
}
}

\begin{document}

\maketitle

\begin{abstract}
Given a pattern $x$ of length $m$ and a text $y$ of length $n$, both over an ordered alphabet, the \emph{order-preserving pattern matching} problem consists in finding all substrings of the text with the same relative order as the pattern.  It is an approximate variant of the well known \emph{exact pattern matching} problem which has gained attention in recent years. This interesting problem finds applications in a lot of fields as time series analysis, like share prices on stock markets, weather data analysis or to musical melody matching.
In this paper we present two new filtering approaches which turn out to be much more effective in practice than the previously presented methods. From our experimental results it turns out that our proposed solutions are up to 2 times faster than the previous solutions reducing the number of false positives up to 99\%.
\end{abstract}

\section{Introduction} \label{sec:introduction}
Given a pattern $x$ of length $m$ and a text $y$ of length $n$, both over a common alphabet $\Sigma$, the \emph{exact string matching problem}  consists in finding all occurrences of the string $x$ in $y$. String matching is a very important subject in the wider domain of text processing and algorithms for the problem are also basic components used in the implementations of practical softwares existing under most operating systems. Moreover, they emphasize programming methods that serve as paradigms in other fields of computer science. Finally they also play an important role in theoretical computer science by providing challenging problems.
The worst case lower bound of the string matching problem is $\bigO(n)$ and was achieved the first time by the well known algorithm by Knuth, Morris and Pratt~\cite{KMP77}. However many string matching algorithms have been also developed to obtain sublinear $\bigO(n\log m/ m)$ performance on average. Among them the Boyer-Moore algorithm~\cite{BM77} deserves a special mention, since it has been particularly successful and has inspired much work.

The \emph{order-preserving pattern matching problem} [2, 3, 8, 9]  (OPPM in short) is an approximate variant of the exact pattern matching problem which has gained attention in recent years. 
In this variant the characters of $x$ and $y$ are drawn from an ordered alphabet $\Sigma$ with a total order relation defined on it.  The task of the problem is to find all substrings of the text with the same relative order as the pattern. 

For instance the relative order of the sequence $x=\langle 6,5,8,4,7 \rangle$ is the sequence $\langle 3,1,0,4,2 \rangle$ since $6$ has rank $3$, $5$ as rank $1$, and so on.
Thus $x$ occurs in the string $y = \langle 8, 11,10,16,15,20,13,17,14,18,20,18,25,17,20,25,26\rangle$ at position $3$, since $x$ and the subsequence $\langle 16,15,20,13,17 \rangle$ share the same relative order.
An other occurrence of $x$ in $y$ is at position $10$ (see Fig.\ref{fig:example}).

\begin{figure}[t!]
\includegraphics[width=0.99\textwidth]{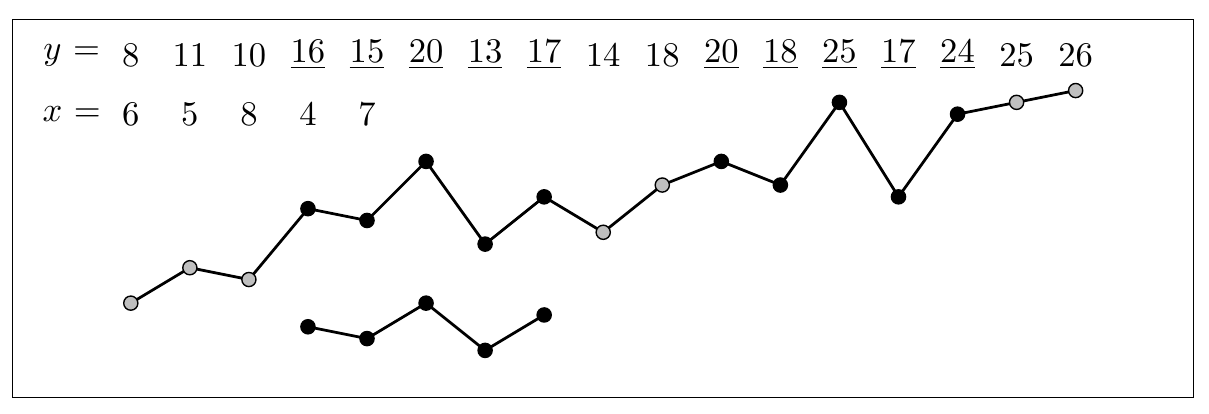}
\label{fig:example}
\caption{Example of a pattern $x$ of length $5$ over an integer alphabet with two order preserving occurrences in a text $y$ of length $17$, at positions $3$ and $10$.}
\end{figure}

The OPPM problem finds applications in the fields where we are interested in finding  patterns affected by relative orders, not by their absolute values. For example, it can be applied to time series analysis like share prices on stock markets, weather data or to musical melody matching of two musical scores.

In the last few years some solutions have been proposed for the order-preserving pattern matching problem. The first solution was presented by Kubica et al.~\cite{KKRRW13} in 2013. They proposed a $\bigO(n+m\log m)$ solution over generic ordered alphabets based on the Knuth-Morris-Pratt algorithm~\cite{KMP77} and a $\bigO(n+m)$ solution in the case of integer alphabets. Some months later Kim et al.~\cite{KEFHIPT14} presented a similar solution running in $\bigO(n+m\log m)$ time based on the KMP approach.  
Although Kim et al. stressed some doubts about the applicability of the Boyer-Moore approach~\cite{BM77} to order-preserving matching problem, in 2013 Cho et al.~\cite{CNPS13}  presented a method for deciding the order-isomorphism between two sequences showing that the Boyer-Moore approach  can be applied also to the order-preserving variant of the pattern matching problem. More recently Chhabra and Tarhio~\cite{CT14} presented a more practical solution based on approximate string matching. Their technique is based on a conversion of the input sequences in binary sequences and on the application of any standard algorithm for exact string matching as a filtration method.

In this paper we present two new families of filtering approaches which turn out to be much more effective in practice than the previously presented methods. While the technique proposed by Chhabra and Tarhio translates the input strings in binary sequences, our methods work on sequences over larger alphabets  in order to speed up the searching process and reduce the number of false positives. From our experimental results it turns out that our proposed solutions are up to 2 times faster than the previous solutions reducing the number of false positives up to 99\% under suitable conditions.

The paper is organized as follows. In Section \ref{sec:definitions} we give preliminary notions and definitions relative to the order-preserving pattern matching problem while in Section \ref{sec:previous} we briefly describe the previous solutions to the problem. Then we present our new solutions in Section \ref{sec:new} and evaluate their performances against the previous algorithms in Section \ref{sec:results}. Conclusions are drawn in Section \ref{sec:conclusions}.

\section{Notions and Basic Definitions} \label{sec:definitions}

A string $x$ over an ordered alphabet $\Sigma$, of size $\sigma$, is defined as a sequence of elements in $\Sigma$. We suppose that a total order relation ``$\leq$'' is defined on it, so that we could establish if $a\leq b$ for each $a,b \in \Sigma$.

We indicate the length of a string $x$ with the symbol $|x|$. We refer to the elements in $x$ with the symbol $x[i]$, for $0\leq i < |x|$. 
Moreover we indicate with $x[i \ldots j]$ the subsequence of $x$ from the element of position $i$ to the element of position $j$ (including the extremes), 
for $0\leq i\leq j < |x|$.

We say that two sequences $x, y \in \Sigma^*$ are order isomorphic if the relative order of their elements is the same. 
More formally we give the following definition.

\begin{definition}[order isomorphism]\label{def:isomorphism}
Given an ordered alphabet $\Sigma$ and two sequences $x,y\in \Sigma^*$ of the same length, we say that $x$ and $y$ are order-isomorphic, and write $x \eq y$, if the following conditions hold
\begin{enumerate}
	\item $|x| = |y|$
	\item $x[i] \leq x[j]$ if and only if $y[i] \leq y[j]$, for $0\leq i,j < |x|$
\end{enumerate} 
\end{definition}

\begin{definition}[rank function]\label{def:order}
Let $x$ be a sequence of length $m$ over an ordered alphabet $\Sigma$.
The rank function of $x$ if a mapping $r : \{0, 1, \ldots, m-1\} \rightarrow \{0,1, \ldots, m-1\}$ such that  $x[r(i)] \leq x[r(j)]$ holds for
each pair  $0\leq i<j<m$. Formally we define
$$
	r(i) = \left| \{ j\ :\ x[j]<x[i] \textrm{ or } (x[j]=x[i] \textrm{ and } j<i) \} \right|
$$
for $0\leq i< m$.
\end{definition}
We will refer to the value $r(i)$ as the \emph{rank} of $x[i]$ in $x$, while we will refer to the sequence $\langle r(0), r(1), \ldots r(m-1)\rangle$ as the \emph{relative order} of $x$.

According to Definition \ref{def:order} we have that  $x[r(0)]$ is the smallest number while $x[r(m-1)]$ is the greater number in $x$. 
If we assume that  $sort(x)$ is the time required to sort all the elements of $x$, then it is easy to observe that the relative order of $x$ can be computed in $\bigO(sort(x))$ time.

In addition, we define the \emph{equality function} of $x$ which indicates which elements of the sequence are equal (if any). More formally we have the following definition.

\begin{definition}[equality function]\label{def:equality}
Let $x$ be a sequence of length $m$ over an ordered alphabet $\Sigma$ and let $r$ be the rank function of $x$. 
The equality function of $x$ if a mapping $eq : \{0, 1, \ldots, m-2\} \rightarrow \{0,1\}$ such that, for each $0\leq i <m$
$$
	eq(i) = \left\{ \begin{array}{ll}
		1 ~~~~~~~~~~~~& \textrm{ if }~~ x[r(i)] = x[r(i+1)]\\
		0 & \textrm{otherwise}
	\end{array} \right.
$$
\end{definition}

Let $r$ be the rank function of a string $x$, such that $m=|x|$, and let $q$ be its equality function. It is easy to prove that $x$ and $y$ are order isomorphic if and only if they share the same rank and equality function, i.e. if and only if the following two conditions hold
\begin{enumerate}
	\item $y[r(i)] \leq y[r(i+1)]$, for $0\leq i< m-1$
	\item $y[r(i)] = y[r(i+1)]$ if and only if $q(i)=1$, for $0\leq i< m-1$
\end{enumerate}

\begin{example}
Let $x=\langle 6, 3, 8, 3, 10, 7, 10 \rangle$ and $y =\langle 2, 1, 4, 1, 5, 3, 5 \rangle$ two sequences of size $7$. We have that the relative order of $x$ is $( 1, 3, 0, 5, 2, 4, 6 )$ while its equality function is $eq(x[i]) = ( 1, 0, 0, 0, 0, 1)$. The two string are order isomorphic according to the definition given above, i.e. $x \eq y$.
\end{example}

\setcounter{instr}{0}
\begin{figure}[!t]
\begin{center}
\begin{tabular}{rl}
\multicolumn{2}{l}{\textsc{Noder-Isomorphism}$(r,eq,y,i)$}\\
\ninstr & \quad \textsf{for $i\leftarrow 0$ to $|x|-1$ do}\\
\ninstr & \quad \quad \textsf{if ($y[r(i)] > y[r(j+i+1)]$) then return false}\\
\ninstr & \quad \quad \textsf{if ($y[r(i)] < y[r(j+i+1)]$ and $eq(i)=1$) then return false}\\
\ninstr & \quad \quad \textsf{if ($y[r(i)] = y[r(j+i+1)]$ and $eq(i)=0$) then return false}\\
\ninstr & \quad \textsf{return true}\\
&\\
\end{tabular}
\caption{The function used to verify if two sequences $x$ and $y[i\ldots i+|x|-1]$ are order isomorphic. We assume that the function receives as input the parameter $r$ and $eq$
which represent the rank function and the equality function of $x$, respectively.}
\label{fig:code0}
\end{center}
\end{figure}

The procedure to verify that two numeric sequences, $x$ and $y$, are order isomorphic is shown in Fig.\ref{fig:code0}. It receives as input the functions $r$ and $q$, computed on $x$ and returns a boolean value indicating if $x\eq y$. The algorithm requires $\bigO(m)$ time, where $m$ is the length of the sequences.
A mismatch occurs when one of the three conditions of lines  2, 3 and 4, holds.

The OPPM problem consists in finding all substring of the text with the same relative order as the pattern. Specifically we have the following formal definition.

\begin{definition}[order preserving pattern matching]
Let $x$ and $y$ be two sequences of length $m$ and $n$, respectively (and $n>m$), both over an ordered alphabet $\Sigma$. The order preserving pattern matching problem consists in finding all indexes $i$, with $0\leq i <n-m$, such that $y[i\ldots i+m-1] \eq x$.
\end{definition}

If an occurrence of the pattern $x$ starts at portion $i$ of the text $y$, we say that $x$ has an order-preserving occurrence at position $i$.

\section{Previous Results} \label{sec:previous}
The OPPM problem has drawn particular attention in the last few years, during which some efficient results have been proposed.

The first algorithm to solve the OPPM problem was presented by Kubica \emph{et al.} in~\cite{KKRRW13}. Their solution was an adaptation of the well Known Knuth-Morris-Pratt algorithm for the exact string matching problem, where the fail function is adapted to compute the order-borders table. The authors proved that this table can be computed in linear time in the length of the pattern $x$, if the relative order of $x$ is known in advance. The overall time complexity of the algorithm is $\bigO(n+m\log m)$, where $m$ is the length of the pattern while $n$ is the length of the text. However in~\cite{CNPS13} Cho \emph{et al.} proved that the algorithm presented in~\cite{KKRRW13} can decide incorrectly when there are equal values in the string. 

The second algorithm based on Knuth-Morris-Pratt was presented later by Kim et al.~\cite{KEFHIPT14}. Their algorithm is based on the prefix representation and it is further optimized according to the nearest neighbor representation. The prefix representation is based on finding the rank of each integer in the prefix. It can be computed easily by inserting each character to the dynamic order statistic tree and then computing the rank of each character in the prefix. 
The time complexity of computing such prefix representation is $O(m\log m)$. The failure function is then computed as in the Knuth-Morris-Pratt algorithm in $O(m\log m)$ time. 
The overall time complexity of this algorithm is $O(n +m\log m)$. Again, this solution does not work properly when there are equal values in the pattern.


The first sublinear solution for the OPPM problem was presented by Cho \emph{et al.} in~\cite{CNPS13}. Their algorithm is an adaptation to OPPM of the well known Boyer-Moore approach.  They apply a $q$-grams technique, i.e. groups of $q$ consecutive characters are treated as a single condensed character, in order to make the shifts longer. 
In this way, a large amount of text can be skipped for long patterns.

More recently Chhabra and Tarhio presented a new practical solution~\cite{CT14} based on a filtration technique. Their algorithm translates the input sequences in two binary sequences  and then use any standard exact pattern matching algorithm as a filtration procedure. 
In particular in their approach a sequence $s$ is translated in a binary sequence $\beta$ of length $|s|-1$ according to the following position
\begin{equation}\label{eq:filter}
	\beta[i] = \left\{ \begin{array}{ll}
			1 & \textrm{ if } s[i] \geq s[i+1]\\
			0 & \textrm{ otherwise }
	\end{array}\right.
\end{equation}
for each $0\leq i < |s|-1$. This translation is unique for a given sequence $s$ and can be performed on line on the text, requiring constant time for each text character.

Thus when a candidate occurrence is found during the filtration phase an additional verification procedure is run in order to check for the order-isomorphism of the candidate substring and the pattern.  Despite its quadratic time complexity, this approach turns out to be simpler and more effective in practice than earlier solutions. It is important to notice that any algorithm for exact string matching can be used as a filtration method. 
The authors also proved that if the underlying filtration algorithm is sublinear and the text is translated on line, the overall complexity of the algorithm is sublinear on average.
Experimental results conducted in~\cite{CT14} show that the filter approach was considerably faster than the algorithm by Cho \emph{et al.}

For the sake of completeness we notice that Crochemore \emph{et al.} presented in~\cite{CIKKLPR13} a solution for the offline version of the OPPM problem based on a new data structure called order-preserving suffix tree. Their solution finds all occurrences of $x$ in $y$ in $\bigO((m \log n)/\log \log m + z)$ where $z$ is the number of occurrences of $x$  in $y$. In this paper we concentrate on the online version of the OPPM problem.

\section{New Efficient Filter Based Algorithms} \label{sec:new}
In this section we present two new general approaches for the OPPM problem. Both of them are based on a filtration technique, as in~\cite{CT14}, but we use information extracted from groups of integers in the input string, as in~\cite{CNPS13}, in order to make the filtration phase more effective in terms of efficiency and accuracy, as discussed below.

\emph{Text filtration} is a largely used technique  in the field of exact and approximate string matching. Specifically, instead of checking at each position of the text if the pattern occurs, it seems to be more efficient to filter text positions and check only when a substring looks like the pattern. When a resemblance has been detected a naive check of the occurrence is performed. In literature filtration techniques are generally  improved by using $q$-grams, i.e. groups of adjacent characters of the string which are considered as a single character of a condensed alphabet.

It is always convenient to use a filtration method which better and faster localize candidate occurrences, which imply accuracy and efficiency of the method, respectively.

The \emph{accuracy} of a filtration method is a value indicating how many false positives are detected during the filtration phase, i.e. the number of candidate occurrences detected by the filtration algorithm which are not real occurrences of the pattern. The \emph{efficiency} is instead related with the time complexity of the procedure we use for managing $q$-grams and with the time efficiency of the overall searching algorithm. It is clear that these two values are strongly related since a low accuracy implies an high number of false positives and, as a consequence, a decrease in the performance of the searching algorithm.

When using $q$-grams, a great accuracy translates in involving greater values of $q$. However, in this context, the value of $q$ represents a trade-off between the computational time required for computing the $q$-grams for each window of the text and the computational time needed for checking false positive candidate occurrences. The larger is the value of $q$, the more time is needed to compute each $q$-gram. On the other hand, the larger is the value of $q$, the smaller is the number of false positives the algorithm finds along the text during the filtration.

In our approaches we make use of the following definition of $q$-neighborhood of an element in an integer string.

\begin{definition}[$q$-neighborhood]
Given a string $x$ of length $m$, we define the $q$-neighborhood of the element $x[i]$, with $0\leq i < m-q$, as the sequence of $q+1$ elements from position $i$ to $i+q$ in $x$, i.e. the sequence $\langle x[i], x[i+1], \ldots, x[i+q] \rangle$.
\end{definition}

Both the filtration methods presented below translate the input sequence in a target numeric sequence which is used for the filtration. Specifically each position $i$ of the sequence is associated with a numeric value computed from the structure of the $q$-neighborhood of the element $x[i]$.


\subsection{The Neighborhood Ranking Approach} \label{sec:nr}
Given a string $x$ of length $m$, we can compute the relative position of the element  $x[i]$ compared with the element $x[j]$ by querying the inequality $x[i]\geq x[j]$. For brevity we will write in symbol $\gte_x(i,j)$ to indicate the boolean value resulting from the above inequality, extending the formal definition given in Equation (\ref{eq:filter}).
Formally we have
\begin{equation}\label{eq:gte}
	\beta_x(i,j) = \left\{ \begin{array}{ll}
			1 & \textrm{ if } x[i] \geq x[j]\\
			0 & \textrm{ otherwise }
	\end{array}\right.
\end{equation}

It is easy to observe that if $\gte_x(i,j)=1$ we have that $r(i)\geq r(j)$ ($x[j]$ precedes $x[i]$ in the ordering of the elements of $x$), otherwise $r(i)< r(j)$.

The neighborhood ranking (\nr) approach associates each position $i$ of the string $x$ (where $0\leq i < m-q$) with the sequence of the relative positions between $x[i]$ and $x[i+j]$, for $j=1,\ldots, q$. In other words we compute the binary sequence $\langle \gte_x(i,i+1), \gte_x(i,i+2), \ldots, \gte_x(i, i+q) \rangle$  of length $q$ indicating the relative positions of the element $x[i]$ compared with other values in its $q$-neighborhood. Of course, we do not include in the sequence the relative position of $\gte(i,i)$, since it doesn't give any additional information.

Since there are $2^{q}$ possible configurations of a binary sequence of length $q$ the string $x$ is converted in a sequence $\chi^q_x$ of length $m-q$, where each element $\chi^q_x[i]$, for $0\leq i<m-q$, is a value  such that $0\leq \chi^q_x[i]< 2^{q}$. 

More formally we have the following definition

\begin{definition}[$q$-NR sequence] \label{def:nr}
Given a string $x$ of length $m$ and an integer $q<m$, the $q$-\nr sequence associated with $x$ is a numeric sequence $\chi_x^q$ of length $m-q$ over the alphabet $\{0, \ldots, 2^q \}$ where
$$
	\chi^q_x[i] =  \sum_{j=1}^q \left( \gte_x(i,i+j)\times 2^{q-j} \right),  \textrm{ for all } 0\leq i < m-q
$$
\end{definition}

\begin{figure}[!t]
\begin{center}
\includegraphics[width=0.89\textwidth]{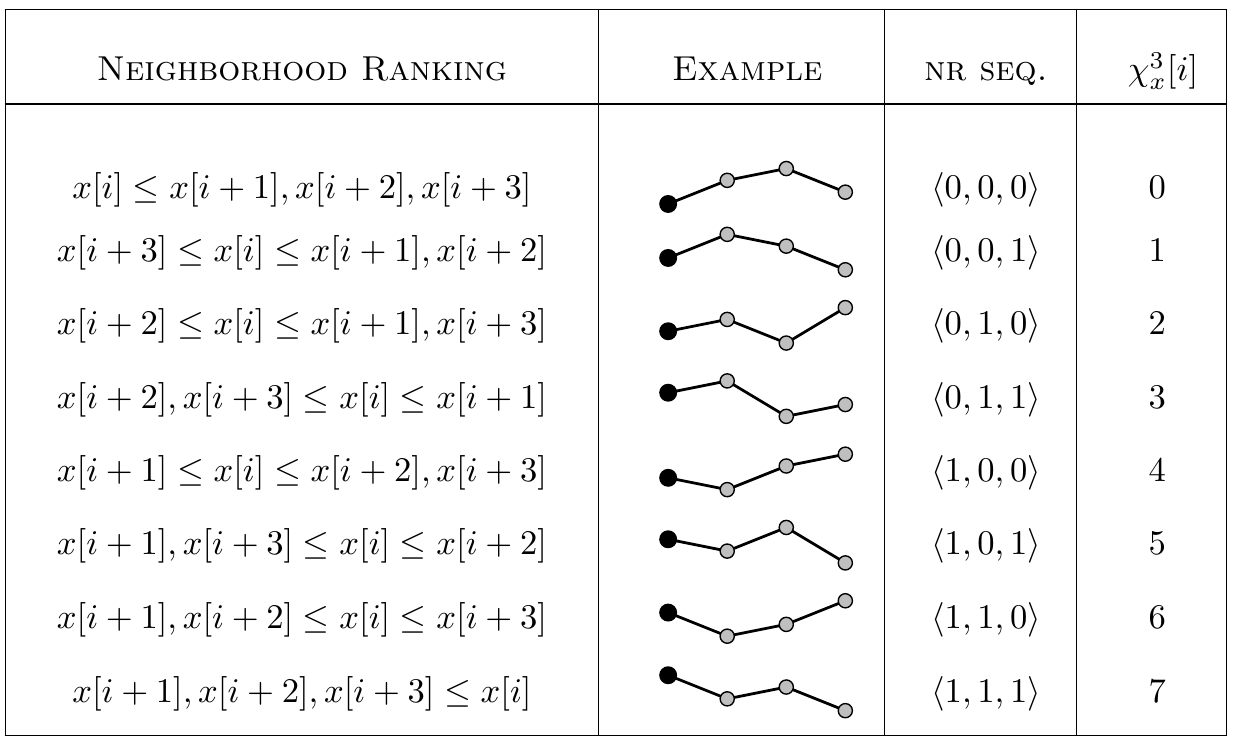}
\caption{The $2^3$ possible $3$-neighborhood ranking sequences associated with element $x[i]$, and their corresponding \nr value.
In the leftmost column we show the ranking position of $x[i]$ compared with other elements in its neighborhood $\langle x[i], x[i+1], x[i+2], x[i+3] \rangle$.}
\label{fig:seq1}
\end{center}
\end{figure}

\begin{example}  \label{ex:ranking}
Let $x = \langle 5, 6, 3, 8, 10, 7, 1, 9, 10, 8 \rangle$ be a sequence of length $10$. The $4$-neighborhood  of the element $x[2]$ is the subsequence $\langle 3, 8, 10, 7, 1\rangle$. Observe that $x[2]$ is greater than $x[6]$ and less than all other values in its $4$-neighborhood. Thus the ranking sequence associated with the element of position $2$ is $\langle 0, 0, 0, 1\rangle$ which translates in a \nr value equal to $1$. In a similar way we can observe that the \nr sequence associated with the element of position $3$ is $\langle 0, 1, 1, 0 \rangle$ which translates in a \nr value equal to $6$. The whole $4$-\nr sequence of length $6$ associated to $x$ is $\chi^4_x = \langle 4, 8, 1, 6, 15, 8\rangle$.
\end{example}

The following Lemma \ref{lem:nr} and Corollary \ref{cor:nr} prove that the \nr approach can be used to filter a text $y$ in order to search for all order preserving occurrences of a pattern $x$. In other words it proves that 
$$\{ i\ |\ x \eq y[i\ldots i+m-1]\} \subseteq \{ i\ |\ \chi^q_x = \chi^q_y[i \ldots i+m-k]\}.$$

\begin{lemma} \label{lem:nr}
Let $x$ and $y$ be two sequences of length $m$ and let $\chi^q_x$ and $\chi^q_y$ the $q$-ranking sequences associated to $x$ and $y$, respectively.
If $x \eq y$ then $\chi^q_x = \chi^q_y$.
\end{lemma}
\begin{proof}
Let $r$ be the rank function associated to $x$ and suppose by hypothesis that $x \eq y$. Then the following statements hold
$$
\begin{array}{rlll}
	\emph{1}. &~&  \textrm{by Definition \ref{def:order} we have } x[r(i)] \leq x[r(i+1)],  \textrm{ for } 0\leq i < m-1;\\
	\emph{2}. && \textrm{by hyphotesis and Def.\ref{def:isomorphism}, } y[r(i)] \leq y[r(i+1)],  \textrm{ for } 0\leq i < m-1;\\
	\emph{3}. && \textrm{then by \emph{1} and \emph{2}, } x[i] \leq x[j]  \textrm{ iff } y[i] \leq y[j],  \textrm{ for } 0\leq i,j < m-1;\\	
	\emph{4}. && \textrm{the previous statement implies that } x[i] \geq x[i+j]  \textrm{ iff } y[i] \geq y[i+j]\\
			&& \textrm{for } 0\leq i < m-q  \textrm{ and }  1\leq j < q;\\	
	\emph{5}. && \textrm{by statement \emph{4}  we have that } \beta_x(i,i+j)  = \beta_y(i,j+j)\\
			&& \textrm{for } 0\leq i < m-q  \textrm{ and }  1\leq j < q;\\	
	\emph{6}. && \textrm{finally, by \emph{5} and Definition \ref{def:nr}, we have } \chi^q_x[i] = \chi^q_y[i],  \textrm{ for } 0\leq i < m-q.
\end{array} 
$$
This last statement proves the thesis.\qeda
\end{proof}

The following corollary prices that the \nr approach can be used as a filtering. It trivially follows from Lemma \ref{lem:nr}.
 
\begin{corollary} \label{cor:nr}
Let $x$ and $y$ be two sequences of length $m$ and $n$, respectively. Let $\chi^q_x$ and $\chi^q_y$ the $q$-ranking sequences associated to $x$ and $y$, respectively.
If $x \eq y[j \ldots j+m-1]$ then $\chi^q_x[i] = \chi^q_y[j+i]$, for $0\leq i < m-q$. \qeda
\end{corollary}

Fig. \ref{fig:code1} shows the procedure used for computing the \nr value associated with the element of the string $x$ at position $i$. The time complexity of the procedure is $\bigO(q)$. Thus, given a pattern $x$ of length $m$, a text $y$ of length $n$ and an integer value $q<m$, we can solve the OPPM problem by searching $\chi^q_y$ for all occurrences of $\chi^q_x$, using any algorithm for the exact string matching problem. During the preprocessing phase we compute the sequence $\chi_x^q$ and the functions $r_x $ and $q_x$. When an occurrence of $\chi^q_x$ is found at position $i$ the verification procedure \textsc{Noder-Isomorphism}$(r,q,y,i)$ (shown in Fig.\ref{fig:code0}) is run in order to check if $x\eq y[i\ldots i+m-1]$.

\setcounter{instr}{0}
\begin{figure}[!t]
\begin{center}
\begin{tabular}{rl}
\multicolumn{2}{l}{\textsc{Compute-NR-Value}$(x,i,q)$}\\
\ninstr & \quad \textsf{$\delta \leftarrow 0$}\\
\ninstr & \quad \textsf{for $j\leftarrow 1$ to $q$ do}\\
\ninstr & \quad \quad \textsf{$\delta = (\delta \ll 1) + \gte_x(i,i+j)$}\\
\ninstr & \quad \textsf{return $\delta$}\\
&\\
\end{tabular}
\caption{The function which computes the $q$-neighborhood ranking value of the element of position $i$ in a sequence $x$. The value id computed in $\bigO(q)$ time.}
\label{fig:code1}
\end{center}
\end{figure}

Since in the worst case the algorithm finds a candidate occurrence at each text position and each verification costs $\bigO(m)$, the worst case time complexity of the algorithm is $\bigO(nm)$, while the filtration phase can be performed with a $\bigO(nq)$ worst case time complexity.
However, following the same analysis of~\cite{CT14}, we easily prove that verification time approaches zero when the length of the pattern grows, so that the filtration time dominates. Thus if the filtration algorithm is sublinear, the total algorithm is sublinear.

\subsection{The Neighborhood Ordering Approach}  \label{sec:no}
The neighborhood ranking approach described in the previous section gives partial information about the relative ordering of the elements in the $q$-neighborhood of an element in $x$. The $q$ binary sequence used to represent each element $x[i]$ is not enough to describe the full ordering information of a set of $q+1$ elements.

The $q$-neighborhood ordering (\no) approach, which we describe in this section, associates each element of the $x$ with a binary sequence which completely describes the ordering disposition of the elements in the $q$-neighborhood of $x[i]$. The number of comparisons we need to order a sequence of $q+1$ elements is between $q$ (the best case) and $q(q+1)/2$ (the worst case). In this latter case it is enough to compare the element $x[j]$, where $i\leq j < i+q$, with each element $x[h]$, where $j<h\leq i+q$. 

Thus each element of position $i$ in $x$, with $0\leq i<m-q$, is associated with a binary sequence of length $q(q+1)/2$ which completely describes the relative order of the susequence $x[ i, \ldots, i+q]$. 
Since there are $(q+1)!$ possible permutations of a set of $q+1$ elements, the string $x$ is converted in a sequence $\varphi_x^q$ of length $m-q$, where each element $\varphi^q_x[i]$ is a value  such that $0\leq \varphi^q_x[i]< q(q+1)/2$.

More formally we have the following definition
\begin{definition}[$q$-NO sequence] \label{def:no}
Given a string $x$ of length $m$ and an integer $q<m$, the $q$-\no sequence associated with $x$ is a numeric sequence $\varphi_x^q$ of length $m-q$ over the alphabet $\{0, \ldots, q(q+1)/2 \}$ where
\begin{equation}
	\varphi^q_x[i] =  \sum_{k=1}^{q} \left( \chi^{k}_x[i+q-k] \times 2^{(k)(k-1)/2}\right),  \textrm{ for all } 0\leq i < m-q
\end{equation}
\end{definition}
Thus the $q$-\no value associated to $x[i]$ is the combination of $q$ different \nr sequences $\chi_x^q[i]$, $\chi_x^{q-1}[i+1]$, \ldots, $\chi_x^1[i+q-1]$. 

For instance the $4$-\no value associated to $x[i]$ is computed as
$$
	\varphi_x^4[i] = \chi_x^4[i] \times 2^{6} + \chi_x^{3}[i+1] \times 2^{2} + \chi_x^{2}[i+2] \times 2 + \chi_x^{1}[i+3]
$$

\begin{figure}[!t]
\begin{center}
\includegraphics[width=0.80\textwidth]{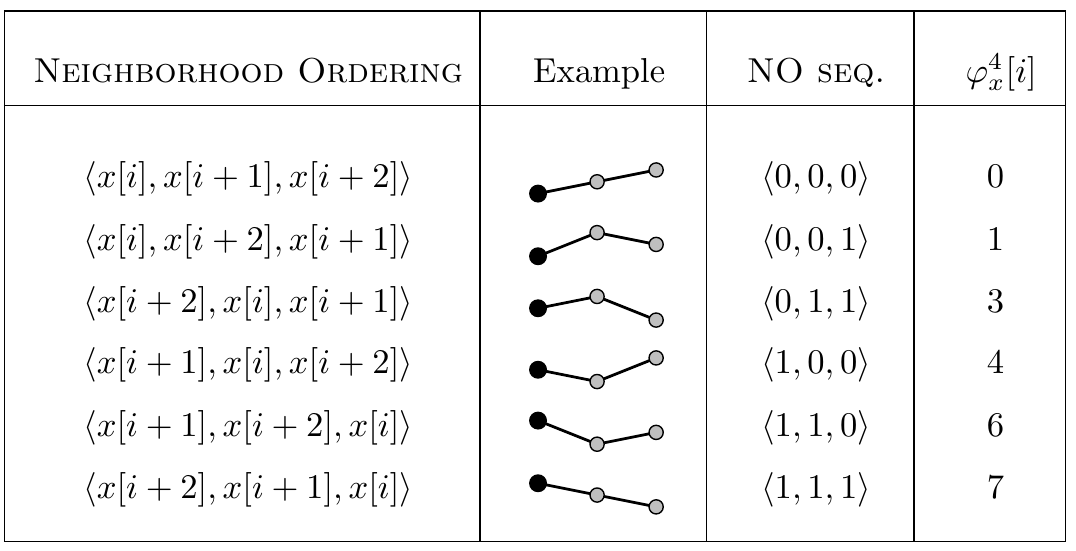}
\caption{The $3!$ possible ordering of the sequence $\langle x[i], x[i+1], x[i+2] \rangle$ and the corresponding binary sequence $\langle \beta_x(i,i+1),  \beta_x(i,i+2),  \beta_x(i+1,i+2)\rangle$.}
\label{fig:ex2}
\end{center}
\end{figure}

\setcounter{instr}{0}
\begin{figure}[!t]
\begin{center}
\begin{tabular}{rl}
\multicolumn{2}{l}{\textsc{Compute-NO-Value}$(x,i,q)$}\\
\ninstr & \quad \textsf{$\delta \leftarrow 0$}\\
\ninstr & \quad \textsf{for $k\leftarrow q$ downto $1$ do}\\
\ninstr & \quad \quad \textsf{for $j \leftarrow 1$ to $k$ do}\\
\ninstr & \quad \quad \quad \textsf{$\delta = (\delta \ll 1) + \beta_x(i+q-k, i+q-k+j)$}\\
\ninstr & \quad \textsf{return $\delta$}\\
&\\
\end{tabular}
\caption{The function which computes the $q$-neighborhood ranking value of the element of position $i$ in a sequence $x$. The value is computed in $\bigO(q^2)$ time.}
\label{fig:code2}
\end{center}
\end{figure}

\begin{example} \label{ex:ordering}
As in \emph{Example \ref{ex:ranking}}, let $x = \langle 5, 6, 3, 8, 10, 7, 1, 9, 10, 8 \rangle$ be a sequence of length $10$. The $3$-neighborhood  of the element $x[3]$ is the subsequence $\langle 8, 10, 7, 1\rangle$. The \no sequence of length $6$ associated with the element of position $2$ is therefore $\langle 0,1,1,1,1,1\rangle$ which translates in a \no value equal to $\varphi_x[3]=31$. In a similar way we can observe that the \nr sequence associated with the element of position $2$ is $\langle 0,0,0,0,1,1 \rangle$ which translates in a \no value equal to $\varphi^4_x[2]=3$. The whole sequence of length $7$ associated to $x$ is $\varphi^4_x = \langle 20, 32, 3, 31, 60, 32, 3\rangle$.
\end{example}

The following Lemma \ref{lem:no} and Corollary \ref{cor:no} prove that the \no approach can be used to filter a text $y$ in order to search for all order preserving occurrences of a pattern $x$. In other words they prove that $$\{ i\ |\ x \eq y[i\ldots i+m-1]\} \subseteq \{ i\ |\ \varphi^q_x = \varphi^q_y[i \ldots i+m-k]\}.$$

\begin{lemma} \label{lem:no}
Let $x$ and $y$ be two sequences of length $m$ and let $\varphi^q_x$ and $\varphi^q_y$ the $q$-ranking sequences associated to $x$ and $y$, respectively.
If $x \eq y$ then $\varphi^q_x = \varphi^q_y$.
\end{lemma}
\begin{proof}
The theorem easily follows from Definition \ref{def:no} and Lemma \ref{lem:nr}. \qeda
\end{proof}

The following corollary proves that the \nr approach can be used as a filtering. It trivially follows from Lemma \ref{lem:no}.
 
\begin{corollary} \label{cor:no}
Let $x$ and $y$ be two sequences of length $m$ and $n$, respectively. Let $\chi^q_x$ and $\chi^q_y$ the $q$-ranking sequences associated to $x$ and $y$, respectively.
If $x \eq y[j \ldots j+m-1]$ then $\chi^q_x[i] = \chi^q_y[j+i]$, for $0\leq i < m-q$. \qeda
\end{corollary}

Fig. \ref{fig:code2} shows the procedure used for computing the \no value associated with the element of the string $x$ at position $i$. The time complexity of the procedure is $\bigO(q^2)$. Thus, given a pattern $x$ of length $m$, a text $y$ of length $n$ and an integer value $q<m$, we can solve the OPPM problem by searching $\varphi^q_y$ for all occurrences of $\varphi^q_x$, using any algorithm for the exact string matching problem. During the preprocessing phase we compute the sequence $\varphi_x^q$ and the functions $r_x $ and $q_x$. When an occurrence of $\varphi^q_x$ is found at position $i$ the verification procedure \textsc{Noder-Isomorphism}$(r,q,y,i)$ (shown in Fig.\ref{fig:code0}) is run in order to check if $x\eq y[i\ldots i+m-1]$.

Also in this case, if the filtration algorithm is sublinear on average, the \no approach has a sublinear behavior on average.

\section{Experimental Evaluations} \label{sec:results}
In this section we present experimental results in order to evaluate the performances of our new filter based algorithms presented in this paper.
In particular we tested our filter approaches against the filter approach of Chhabra and Tarhio~\cite{CT14}, which is, to the best of our knowledge, the most effective solution in practical cases.
In the experimental evaluation conducted in~\cite{CT14} the \textsc{sbndm$2$}  and \textsc{sbndm$4$} algorithms~\cite{DHPT10} turned out to be the most effective exact string matching algorithms which can be used in combination with the filter technique. Following the same line, in our experimental evaluation we use in all cases the \textsc{sbndm$2$} algorithm. However any other exact string matching algorithm could be used for this purpose.
In our dataset we use the following names to identify the tested algorithms
\begin{itemize}
	\item \textsc{Fct}: the \textsc{sbndm$2$} algorithm based on the filter approach by Chhabra and Jorma Tarhio presented in~\cite{CT14};
	\item \textsc{Nr$q$}: the \textsc{sbndm$2$} algorithm based on the Neighborhood Ranking approach presented in Section \ref{sec:nr}
	\item \textsc{No$q$}: the \textsc{sbndm$2$} algorithm based on the Neighborhood Ordering approach presented in Section \ref{sec:no}
\end{itemize}

We do not compare our solution with the Boyer-Moor approach by Cho \emph{et al.}~\cite{CNPS13} since it was shown to be less efficient than the algorithm by Chhabra and Tarhio in all cases.
We evaluated our filter based solutions in terms of efficiency, i.e. the running times, and accuracy, i.e. the percentage of false positives detected during the filtration phase. In particular for the \textsc{Fct} algorithm we will report the average running times, in milliseconds, and the average number of false positives detected every $2^{20}$ text characters.
Instead, for all other algorithms in the set, we will report the following two values 
\begin{itemize}
	\item the speed up of the running times obtained when compared with the time used by the \textsc{Fct} algorithm. If \emph{time}$(\textsc{Fct})$ is the running time of the \textsc{Fct} algorithm and $t$ is the running time of our algorithm, then the speed up is computed as \emph{time}$(\textsc{Fct})/t$.
	\item the percentage of the gain in the number of false positives detected by the algorithm when compared with the \textsc{Fct} algorithm. If \emph{fp}$(\textsc{Fct})$ is the number of false positives detected on average by the \textsc{Fct} algorithm and $fp$ is number of false positives detected by our filter approach, then the gain is computed as $(100 \times (\emph{fp}(\textsc{Fct})-fp)/\emph{fp}(\textsc{Fct})$.
\end{itemize}

We tested our solutions on sequences of short integer values (each element is an integer in the range $[0\ldots 256]$), long integer values (where each element is an integer in the range $[0\ldots 10.000]$) and floating point values (each element is a floating point in the range $[0.0 \ldots 10000.99]$). However we don't observe sensible differences in the results, thus in the following table we report for brevity the results obtained on short integer sequences. All texts have $1$ million of  elements. In particular we tested our algorithm on the following set of short integer sequences.
\begin{itemize}
\item \textsc{Rand-$\delta$}: a sequence of random integer values ringing around a fixed mean equal to $100$. Each value of the sequence is randomly chosen around the mean with a variability of $\delta$, so that the text can be seen as a random sequence of integers between $100-\delta$ and $100+\delta$ with a uniform distribution.
\item \textsc{Period-$\delta$}: a sequence of random integer values ringing around a periodic function with a period of $10$ elements. Each value of the sequence is randomly chosen around the function with a variability of $\delta$. All values of the sequences are always in the range $\{0\ldots 200+\delta\}$.
\end{itemize}

For each text in the set we randomly select $100$ patterns extracted from the text and compute the average running time over the $100$ runs. We also computed the average number of false positives detected by the algorithms during the search.
All the algorithms have been implemented using the \textsc{C} programming language and have been compiled on an MacBook Pro using the $\texttt{gcc}$ compiler Apple LLVM version 5.1 (based on LLVM 3.4svn) with $8$Gb Ram. During the compilation we use the \texttt{-O3} optimization option.

In the following table running times are expressed in milliseconds. Best results have been underlined.


\begin{table}[t]
\begin{scriptsize}
\begin{tabular*}{0.99\textwidth}{@{\extracolsep{\fill}}l|c|ccccc|ccc}
\hline
\hline
&&&&&&&&&\\
~$m$ ~~~~~&~ \textsc{Fct} ~&~ \textsc{Nr2} ~&~ \textsc{Nr3} ~&~ \textsc{Nr4} ~&~ \textsc{Nr5} ~&~ 
\textsc{Nr6} ~~~&~ \textsc{No2} ~&~ \textsc{No3} ~&~ \textsc{No4} ~\\[0.2cm]
\hline
 8  &  \cellcolor[gray]{0.9}\textsf{44.29}  &  \textsf{1.16}  &  \textsf{1.28}  &  \textsf{1.25}  &  \textsf{1.25}  &  \textsf{1.24}  & \best{\textsf{1.89}}  &  \textsf{1.71}  &  \textsf{1.11}  \\
 12  &  \cellcolor[gray]{0.9}\textsf{28.39}  &  \textsf{1.16}  &  \textsf{1.37}  &  \textsf{1.37}  &  \textsf{1.33}  &  \textsf{1.19}  &  \textsf{1.64}  & \best{\textsf{2.00}}  &  \textsf{1.64}  \\
 16  &  \cellcolor[gray]{0.9}\textsf{20.65}  &  \textsf{1.15}  &  \textsf{1.30}  &  \textsf{1.43}  &  \textsf{1.34}  &  \textsf{1.14}  &  \textsf{1.42}  & \best{\textsf{2.01}}  &  \textsf{1.83}  \\
 20  &  \cellcolor[gray]{0.9}\textsf{16.29}  &  \textsf{1.15}  &  \textsf{1.30}  &  \textsf{1.45}  &  \textsf{1.41}  &  \textsf{1.14}  &  \textsf{1.39}  & \best{\textsf{2.00}}  &  \textsf{1.93}  \\
 24  &  \cellcolor[gray]{0.9}\textsf{13.64}  &  \textsf{1.16}  &  \textsf{1.29}  &  \textsf{1.42}  &  \textsf{1.44}  &  \textsf{1.12}  &  \textsf{1.34}  &  \textsf{1.91}  & \best{\textsf{2.01}}  \\
 28  &  \cellcolor[gray]{0.9}\textsf{11.48}  &  \textsf{1.16}  &  \textsf{1.28}  &  \textsf{1.44}  &  \textsf{1.45}  &  \textsf{1.11}  &  \textsf{1.31}  &  \textsf{1.88}  & \best{\textsf{1.96}}  \\
 32  &  \cellcolor[gray]{0.9}\textsf{10.34}  &  \textsf{1.18}  &  \textsf{1.30}  &  \textsf{1.40}  &  \textsf{1.46}  &  \textsf{1.12}  &  \textsf{1.30}  &  \textsf{1.83}  & \best{\textsf{2.05}}  \\
\hline
\hline
 8  &  \cellcolor[gray]{0.9}\textsf{15713.46}  &  \textsf{84.1}  &  \textsf{92.4}  &  \textsf{95.1}  &  \textsf{94.0}  &  \textsf{90.2}  &  \textsf{97.5}  &  \textsf{99.1}  & \best{\textsf{99.6}}  \\
 12  &  \cellcolor[gray]{0.9}\textsf{1420.78}  &  \textsf{95.8}  &  \textsf{99.3}  &  \textsf{99.7}  &  \textsf{99.8}  &  \textsf{97.5}  &  \textsf{99.8}  &  \textsf{100.0}  & \best{\textsf{100.0}}  \\
 16  &  \cellcolor[gray]{0.9}\textsf{123.22}  &  \textsf{99.4}  & \best{\textsf{100.0}}  & \best{\textsf{100.0}}  & \best{\textsf{100.0}}  &  \textsf{99.7}  & \best{\textsf{100.0}}  & \best{\textsf{100.0}}  & \best{\textsf{100.0}}  \\
 20  &  \cellcolor[gray]{0.9}\textsf{12.07}  & \best{\textsf{100.0}}  & \best{\textsf{100.0}}  & \best{\textsf{100.0}}  & \best{\textsf{100.0}}  & \best{\textsf{100.0}}  & \best{\textsf{100.0}}  & \best{\textsf{100.0}}  & \best{\textsf{100.0}}  \\
 24  &  \cellcolor[gray]{0.9}\textsf{1.01}  & \best{\textsf{100.0}}  & \best{\textsf{100.0}}  & \best{\textsf{100.0}}  & \best{\textsf{100.0}}  & \best{\textsf{100.0}}  & \best{\textsf{100.0}}  & \best{\textsf{100.0}}  & \best{\textsf{100.0}}  \\
 28  &  \cellcolor[gray]{0.9}\textsf{0.02}  & \best{\textsf{100.0}}  & \best{\textsf{100.0}}  & \best{\textsf{100.0}}  & \best{\textsf{100.0}}  & \best{\textsf{100.0}}  & \best{\textsf{100.0}}  & \best{\textsf{100.0}}  & \best{\textsf{100.0}}  \\
 32  &  \cellcolor[gray]{0.9}\textsf{0.00}  & -  & -  & -  & -  & -  & -  & -  &- \\
 \hline
\hline
\end{tabular*}\\[0.2cm]
\end{scriptsize}
\caption{\label{tab:rand1}Experimental results on a \textsc{Rand-5} short integer sequence.}
\end{table}

\begin{table}[t]
\begin{scriptsize}
\begin{tabular*}{0.99\textwidth}{@{\extracolsep{\fill}}l|c|ccccc|ccc}
\hline
\hline
&&&&&&&&&\\
~$m$ ~~~~~&~ \textsc{Fct} ~&~ \textsc{Nr2} ~&~ \textsc{Nr3} ~&~ \textsc{Nr4} ~&~ \textsc{Nr5} ~&~ 
\textsc{Nr6} ~~~&~ \textsc{No2} ~&~ \textsc{No3} ~&~ \textsc{No4} ~\\[0.2cm]
\hline
 8  &  \cellcolor[gray]{0.9}\textsf{42.34}  &  \textsf{1.13}  &  \textsf{1.27}  &  \textsf{1.25}  &  \textsf{1.26}  &  \textsf{1.22}  & \best{\textsf{1.92}}  &  \textsf{1.68}  &  \textsf{1.08}  \\
 12  &  \cellcolor[gray]{0.9}\textsf{27.93}  &  \textsf{1.17}  &  \textsf{1.40}  &  \textsf{1.37}  &  \textsf{1.32}  &  \textsf{1.21}  &  \textsf{1.71}  & \best{\textsf{2.04}}  &  \textsf{1.63}  \\
 16  &  \cellcolor[gray]{0.9}\textsf{20.05}  &  \textsf{1.15}  &  \textsf{1.32}  &  \textsf{1.41}  &  \textsf{1.33}  &  \textsf{1.15}  &  \textsf{1.48}  & \best{\textsf{2.04}}  &  \textsf{1.81}  \\
 20  &  \cellcolor[gray]{0.9}\textsf{15.85}  &  \textsf{1.15}  &  \textsf{1.29}  &  \textsf{1.42}  &  \textsf{1.37}  &  \textsf{1.11}  &  \textsf{1.38}  & \best{\textsf{2.00}}  &  \textsf{1.90}  \\
 24  &  \cellcolor[gray]{0.9}\textsf{13.31}  &  \textsf{1.17}  &  \textsf{1.31}  &  \textsf{1.47}  &  \textsf{1.42}  &  \textsf{1.12}  &  \textsf{1.36}  &  \textsf{1.99}  & \best{\textsf{2.02}}  \\
 28  &  \cellcolor[gray]{0.9}\textsf{11.38}  &  \textsf{1.17}  &  \textsf{1.31}  &  \textsf{1.42}  &  \textsf{1.45}  &  \textsf{1.09}  &  \textsf{1.35}  &  \textsf{1.94}  & \best{\textsf{2.07}}  \\
 32  &  \cellcolor[gray]{0.9}\textsf{9.96}  &  \textsf{1.16}  &  \textsf{1.29}  &  \textsf{1.45}  &  \textsf{1.46}  &  \textsf{1.09}  &  \textsf{1.29}  &  \textsf{1.87}  & \best{\textsf{2.09}}  \\
\hline
\hline
 8  &  \cellcolor[gray]{0.9}\textsf{14326.78}  &  \textsf{83.6}  &  \textsf{92.3}  &  \textsf{95.6}  &  \textsf{92.9}  &  \textsf{90.2}  &  \textsf{97.7}  &  \textsf{99.3}  & \best{\textsf{99.7}}  \\
 12  &  \cellcolor[gray]{0.9}\textsf{1295.88}  &  \textsf{96.4}  &  \textsf{99.5}  &  \textsf{99.9}  &  \textsf{99.9}  &  \textsf{97.8}  &  \textsf{99.9}  &  \textsf{100.0}  & \best{\textsf{100.0}}  \\
 16  &  \cellcolor[gray]{0.9}\textsf{118.79}  &  \textsf{99.3}  & \best{\textsf{100.0}}  & \best{\textsf{100.0}}  & \best{\textsf{100.0}}  &  \textsf{99.7}  & \best{\textsf{100.0}}  & \best{\textsf{100.0}}  & \best{\textsf{100.0}}  \\
 20  &  \cellcolor[gray]{0.9}\textsf{10.43}  & \best{\textsf{100.0}}  & \best{\textsf{100.0}}  & \best{\textsf{100.0}}  & \best{\textsf{100.0}}  & \best{\textsf{100.0}}  & \best{\textsf{100.0}}  & \best{\textsf{100.0}}  & \best{\textsf{100.0}}  \\
 24  &  \cellcolor[gray]{0.9}\textsf{0.71}  & \best{\textsf{100.0}}  & \best{\textsf{100.0}}  & \best{\textsf{100.0}}  & \best{\textsf{100.0}}  & \best{\textsf{100.0}}  & \best{\textsf{100.0}}  & \best{\textsf{100.0}}  & \best{\textsf{100.0}}  \\
 28  &  \cellcolor[gray]{0.9}\textsf{0.00}  & -  & -  & -  & -  & -  & -  & -  &- \\
 32  &  \cellcolor[gray]{0.9}\textsf{0.00}  & -  & -  & -  & -  & -  & -  & -  &- \\
 \hline
\hline
\end{tabular*}\\[0.2cm]
\end{scriptsize}
\caption{\label{tab:rand2}Experimental results on a \textsc{Rand-20} short integer sequence.}
\end{table}

\begin{table}[t]
\begin{scriptsize}
\begin{tabular*}{0.99\textwidth}{@{\extracolsep{\fill}}l|c|ccccc|ccc}
\hline
\hline
&&&&&&&&&\\
~$m$ ~~~~~&~ \textsc{Fct} ~&~ \textsc{Nr2} ~&~ \textsc{Nr3} ~&~ \textsc{Nr4} ~&~ \textsc{Nr5} ~&~ 
\textsc{Nr6} ~~~&~ \textsc{No2} ~&~ \textsc{No3} ~&~ \textsc{No4} ~\\[0.2cm]
\hline
 8  &  \cellcolor[gray]{0.9}\textsf{42.62}  &  \textsf{1.16}  &  \textsf{1.28}  &  \textsf{1.28}  &  \textsf{1.25}  &  \textsf{1.25}  & \best{\textsf{1.94}}  &  \textsf{1.70}  &  \textsf{1.09}  \\
 12  &  \cellcolor[gray]{0.9}\textsf{28.35}  &  \textsf{1.19}  &  \textsf{1.41}  &  \textsf{1.39}  &  \textsf{1.36}  &  \textsf{1.21}  &  \textsf{1.75}  & \best{\textsf{2.06}}  &  \textsf{1.65}  \\
 16  &  \cellcolor[gray]{0.9}\textsf{20.37}  &  \textsf{1.18}  &  \textsf{1.32}  &  \textsf{1.44}  &  \textsf{1.37}  &  \textsf{1.17}  &  \textsf{1.49}  & \best{\textsf{2.09}}  &  \textsf{1.83}  \\
 20  &  \cellcolor[gray]{0.9}\textsf{16.12}  &  \textsf{1.15}  &  \textsf{1.29}  &  \textsf{1.46}  &  \textsf{1.39}  &  \textsf{1.12}  &  \textsf{1.39}  & \best{\textsf{2.04}}  &  \textsf{1.95}  \\
 24  &  \cellcolor[gray]{0.9}\textsf{13.35}  &  \textsf{1.18}  &  \textsf{1.30}  &  \textsf{1.46}  &  \textsf{1.44}  &  \textsf{1.13}  &  \textsf{1.36}  &  \textsf{1.97}  & \best{\textsf{1.99}}  \\
 28  &  \cellcolor[gray]{0.9}\textsf{11.60}  &  \textsf{1.18}  &  \textsf{1.32}  &  \textsf{1.47}  &  \textsf{1.50}  &  \textsf{1.14}  &  \textsf{1.37}  &  \textsf{1.96}  & \best{\textsf{2.06}}  \\
 32  &  \cellcolor[gray]{0.9}\textsf{10.06}  &  \textsf{1.16}  &  \textsf{1.29}  &  \textsf{1.45}  &  \textsf{1.48}  &  \textsf{1.10}  &  \textsf{1.33}  &  \textsf{1.89}  & \best{\textsf{2.07}}  \\
\hline
\hline
 8  &  \cellcolor[gray]{0.9}\textsf{15413.57}  &  \textsf{86.6}  &  \textsf{93.7}  &  \textsf{95.9}  &  \textsf{94.4}  &  \textsf{91.9}  &  \textsf{98.1}  &  \textsf{99.4}  & \best{\textsf{99.8}}  \\
 12  &  \cellcolor[gray]{0.9}\textsf{1492.39}  &  \textsf{97.0}  &  \textsf{99.6}  &  \textsf{99.9}  &  \textsf{99.9}  &  \textsf{98.1}  &  \textsf{99.9}  &  \textsf{100.0}  & \best{\textsf{100.0}}  \\
 16  &  \cellcolor[gray]{0.9}\textsf{114.82}  &  \textsf{99.3}  & \best{\textsf{100.0}}  & \best{\textsf{100.0}}  & \best{\textsf{100.0}}  &  \textsf{99.7}  & \best{\textsf{100.0}}  & \best{\textsf{100.0}}  & \best{\textsf{100.0}}  \\
 20  &  \cellcolor[gray]{0.9}\textsf{9.83}  & \best{\textsf{100.0}}  & \best{\textsf{100.0}}  & \best{\textsf{100.0}}  & \best{\textsf{100.0}}  & \best{\textsf{100.0}}  & \best{\textsf{100.0}}  & \best{\textsf{100.0}}  & \best{\textsf{100.0}}  \\
 24  &  \cellcolor[gray]{0.9}\textsf{0.83}  & \best{\textsf{100.0}}  & \best{\textsf{100.0}}  & \best{\textsf{100.0}}  & \best{\textsf{100.0}}  & \best{\textsf{100.0}}  & \best{\textsf{100.0}}  & \best{\textsf{100.0}}  & \best{\textsf{100.0}}  \\
 28  &  \cellcolor[gray]{0.9}\textsf{0.00}  & -  & -  & -  & -  & -  & -  & -  &- \\
 32  &  \cellcolor[gray]{0.9}\textsf{0.00}  & -  & -  & -  & -  & -  & -  & -  &- \\
 \hline
\hline
\end{tabular*}\\[0.2cm]
\end{scriptsize}
\caption{\label{tab:rand3}Experimental results on a \textsc{Rand-40} short integer sequence.}
\end{table}

\subsection*{Experimental Results on Random Sequences}
Experimental results on \textsc{Rand}-$\delta$ numeric sequences have been conducted with values of $\delta=5, 20, 40$ (see Table \ref{tab:rand1},  Table \ref{tab:rand2} and  Table \ref{tab:rand3}). The results show as the \textsc{No} approach is the best choice in all cases, achieving a speed up of $2.0$ if compared with the \textsc{Fct} algorithm. Also the \textsc{Nr} approach achieves always a good speed up which is between $1.15$ and $1.50$.
The gain in number of detected false positives is impressive and is in most cases between $90\%$ and $100\%$.
It is interesting to observe also that the value of $\delta$ do not affect the running times and the number of false positives detected during the search, which are very similar in the three tables.



\begin{table}[t]
\begin{scriptsize}
\begin{tabular*}{0.99\textwidth}{@{\extracolsep{\fill}}l|c|ccccc|ccc}
\hline
\hline
&&&&&&&&&\\
~$m$ ~~~~~&~ \textsc{Fct} ~&~ \textsc{Nr2} ~&~ \textsc{Nr3} ~&~ \textsc{Nr4} ~&~ \textsc{Nr5} ~&~ 
\textsc{Nr6} ~~~&~ \textsc{No2} ~&~ \textsc{No3} ~&~ \textsc{No4} ~\\[0.2cm]
\hline
 8  &  \cellcolor[gray]{0.9}\textsf{41.08}  &  \textsf{0.99}  & \best{\textsf{1.05}}  &  \textsf{0.88}  &  \textsf{0.79}  &  \textsf{0.90}  &  \textsf{0.88}  &  \textsf{0.73}  &  \textsf{0.60}  \\
 12  &  \cellcolor[gray]{0.9}\textsf{36.42}  & \best{\textsf{1.06}}  &  \textsf{1.02}  &  \textsf{0.94}  &  \textsf{0.86}  &  \textsf{0.91}  &  \textsf{0.81}  &  \textsf{0.67}  &  \textsf{0.69}  \\
 16  &  \cellcolor[gray]{0.9}\textsf{34.03}  & \best{\textsf{1.04}}  &  \textsf{0.86}  &  \textsf{0.78}  &  \textsf{0.74}  &  \textsf{1.00}  &  \textsf{0.77}  &  \textsf{0.64}  &  \textsf{0.60}  \\
 20  &  \cellcolor[gray]{0.9}\textsf{35.31}  & \best{\textsf{0.98}}  &  \textsf{0.89}  &  \textsf{0.88}  &  \textsf{0.84}  &  \textsf{0.92}  &  \textsf{0.73}  &  \textsf{0.60}  &  \textsf{0.55}  \\
 24  &  \cellcolor[gray]{0.9}\textsf{37.90}  & \best{\textsf{1.34}}  &  \textsf{1.33}  &  \textsf{1.30}  &  \textsf{1.18}  &  \textsf{1.15}  &  \textsf{0.99}  &  \textsf{0.82}  &  \textsf{0.76}  \\
 28  &  \cellcolor[gray]{0.9}\textsf{36.26}  & \best{\textsf{1.17}}  &  \textsf{1.09}  &  \textsf{1.10}  &  \textsf{1.04}  &  \textsf{0.97}  &  \textsf{0.78}  &  \textsf{0.64}  &  \textsf{0.56}  \\
 32  &  \cellcolor[gray]{0.9}\textsf{35.38}  &  \textsf{1.10}  & \best{\textsf{1.15}}  &  \textsf{1.05}  &  \textsf{0.95}  &  \textsf{0.94}  &  \textsf{0.82}  &  \textsf{0.65}  &  \textsf{0.59}  \\
\hline
\hline
 8  &  \cellcolor[gray]{0.9}\textsf{48697.90}  &  \textsf{78.1}  &  \textsf{78.2}  &  \textsf{75.0}  &  \textsf{61.8}  &  \textsf{83.0}  &  \textsf{89.1}  &  \textsf{94.2}  & \best{\textsf{95.8}}  \\
 12  &  \cellcolor[gray]{0.9}\textsf{45427.73}  &  \textsf{66.4}  &  \textsf{72.8}  &  \textsf{74.6}  &  \textsf{71.4}  &  \textsf{67.7}  &  \textsf{76.3}  &  \textsf{82.4}  & \best{\textsf{84.9}}  \\
 16  &  \cellcolor[gray]{0.9}\textsf{32091.18}  &  \textsf{54.1}  &  \textsf{63.6}  &  \textsf{66.0}  &  \textsf{63.9}  &  \textsf{55.3}  &  \textsf{66.3}  &  \textsf{72.7}  & \best{\textsf{74.5}}  \\
 20  &  \cellcolor[gray]{0.9}\textsf{26337.31}  &  \textsf{41.0}  &  \textsf{49.0}  &  \textsf{52.6}  &  \textsf{53.6}  &  \textsf{43.0}  &  \textsf{53.4}  &  \textsf{59.1}  & \best{\textsf{61.5}}  \\
 24  &  \cellcolor[gray]{0.9}\textsf{23100.22}  &  \textsf{42.3}  &  \textsf{56.9}  &  \textsf{61.6}  &  \textsf{62.5}  &  \textsf{44.0}  &  \textsf{60.5}  &  \textsf{66.7}  & \best{\textsf{69.6}}  \\
 28  &  \cellcolor[gray]{0.9}\textsf{23296.19}  &  \textsf{53.2}  &  \textsf{63.0}  &  \textsf{70.7}  &  \textsf{73.1}  &  \textsf{55.1}  &  \textsf{65.8}  &  \textsf{73.8}  & \best{\textsf{76.7}}  \\
 32  &  \cellcolor[gray]{0.9}\textsf{17959.33}  &  \textsf{49.7}  &  \textsf{66.6}  &  \textsf{72.0}  &  \textsf{75.6}  &  \textsf{50.5}  &  \textsf{68.9}  &  \textsf{75.2}  & \best{\textsf{79.4}}  \\
 \hline
\hline
\end{tabular*}\\[0.2cm]
\end{scriptsize}
\caption{\label{periodic1}Experimental results on a \textsc{Period-5} short integer sequence.}
\end{table}

\begin{table}[t]
\begin{scriptsize}
\begin{tabular*}{0.99\textwidth}{@{\extracolsep{\fill}}l|c|ccccc|ccc}
\hline
\hline
&&&&&&&&&\\
~$m$ ~~~~~&~ \textsc{Fct} ~&~ \textsc{Nr2} ~&~ \textsc{Nr3} ~&~ \textsc{Nr4} ~&~ \textsc{Nr5} ~&~ 
\textsc{Nr6} ~~~&~ \textsc{No2} ~&~ \textsc{No3} ~&~ \textsc{No4} ~\\[0.2cm]
\hline
 8  &  \cellcolor[gray]{0.9}\textsf{42.35}  &  \textsf{0.98}  & \best{\textsf{1.18}}  &  \textsf{0.91}  &  \textsf{0.81}  &  \textsf{0.89}  &  \textsf{1.02}  &  \textsf{0.83}  &  \textsf{0.68}  \\
 12  &  \cellcolor[gray]{0.9}\textsf{39.09}  &  \textsf{1.11}  & \best{\textsf{1.14}}  &  \textsf{1.06}  &  \textsf{0.98}  &  \textsf{1.00}  &  \textsf{1.02}  &  \textsf{0.88}  &  \textsf{0.93}  \\
 16  &  \cellcolor[gray]{0.9}\textsf{34.25}  & \best{\textsf{1.11}}  &  \textsf{1.01}  &  \textsf{1.02}  &  \textsf{1.01}  &  \textsf{1.08}  &  \textsf{0.96}  &  \textsf{0.87}  &  \textsf{0.87}  \\
 20  &  \cellcolor[gray]{0.9}\textsf{35.41}  &  \textsf{1.10}  &  \textsf{1.09}  &  \textsf{1.21}  & \best{\textsf{1.21}}  &  \textsf{1.07}  &  \textsf{0.97}  &  \textsf{0.89}  &  \textsf{0.89}  \\
 24  &  \cellcolor[gray]{0.9}\textsf{35.15}  &  \textsf{1.31}  &  \textsf{1.51}  & \best{\textsf{1.67}}  &  \textsf{1.60}  &  \textsf{1.14}  &  \textsf{1.15}  &  \textsf{1.10}  &  \textsf{1.18}  \\
 28  &  \cellcolor[gray]{0.9}\textsf{32.23}  &  \textsf{1.23}  &  \textsf{1.40}  & \best{\textsf{1.56}}  &  \textsf{1.36}  &  \textsf{1.07}  &  \textsf{1.04}  &  \textsf{1.08}  &  \textsf{1.15}  \\
 32  &  \cellcolor[gray]{0.9}\textsf{30.34}  &  \textsf{1.43}  & \best{\textsf{1.60}}  &  \textsf{1.53}  &  \textsf{1.43}  &  \textsf{1.22}  &  \textsf{1.19}  &  \textsf{1.11}  &  \textsf{1.07}  \\
\hline
\hline
 8  &  \cellcolor[gray]{0.9}\textsf{62122.44}  &  \textsf{56.9}  &  \textsf{77.8}  &  \textsf{71.5}  &  \textsf{57.1}  &  \textsf{60.9}  &  \textsf{84.7}  &  \textsf{91.8}  & \best{\textsf{95.9}}  \\
 12  &  \cellcolor[gray]{0.9}\textsf{50264.79}  &  \textsf{56.5}  &  \textsf{72.8}  &  \textsf{77.3}  &  \textsf{76.7}  &  \textsf{58.6}  &  \textsf{77.0}  &  \textsf{85.0}  & \best{\textsf{88.1}}  \\
 16  &  \cellcolor[gray]{0.9}\textsf{32026.85}  &  \textsf{60.0}  &  \textsf{73.8}  &  \textsf{79.4}  &  \textsf{80.5}  &  \textsf{62.4}  &  \textsf{78.8}  &  \textsf{86.3}  & \best{\textsf{89.2}}  \\
 20  &  \cellcolor[gray]{0.9}\textsf{23138.04}  &  \textsf{61.1}  &  \textsf{77.4}  &  \textsf{83.2}  &  \textsf{86.3}  &  \textsf{63.3}  &  \textsf{81.2}  &  \textsf{87.8}  & \best{\textsf{91.3}}  \\
 24  &  \cellcolor[gray]{0.9}\textsf{16535.75}  &  \textsf{65.1}  &  \textsf{82.8}  &  \textsf{88.6}  &  \textsf{91.0}  &  \textsf{68.0}  &  \textsf{85.3}  &  \textsf{91.3}  & \best{\textsf{94.2}}  \\
 28  &  \cellcolor[gray]{0.9}\textsf{12181.13}  &  \textsf{72.7}  &  \textsf{85.4}  &  \textsf{92.7}  &  \textsf{94.9}  &  \textsf{74.8}  &  \textsf{88.8}  &  \textsf{94.8}  & \best{\textsf{96.8}}  \\
 32  &  \cellcolor[gray]{0.9}\textsf{9276.84}  &  \textsf{75.2}  &  \textsf{90.4}  &  \textsf{94.2}  &  \textsf{97.0}  &  \textsf{76.1}  &  \textsf{91.4}  &  \textsf{95.4}  & \best{\textsf{98.0}}  \\
 \hline
\hline
\end{tabular*}\\[0.2cm]
\end{scriptsize}
\caption{\label{periodic2}Experimental results on a \textsc{Period-20} short integer sequence.}
\end{table}

\begin{table}[t]
\begin{scriptsize}
\begin{tabular*}{0.99\textwidth}{@{\extracolsep{\fill}}l|c|ccccc|ccc}
\hline
\hline
&&&&&&&&&\\
~$m$ ~~~~~&~ \textsc{Fct} ~&~ \textsc{Nr2} ~&~ \textsc{Nr3} ~&~ \textsc{Nr4} ~&~ \textsc{Nr5} ~&~ 
\textsc{Nr6} ~~~&~ \textsc{No2} ~&~ \textsc{No3} ~&~ \textsc{No4} ~\\[0.2cm]
\hline
 8  &  \cellcolor[gray]{0.9}\textsf{45.07}  &  \textsf{0.93}  & \best{\textsf{1.18}}  &  \textsf{0.94}  &  \textsf{0.81}  &  \textsf{0.89}  &  \textsf{1.12}  &  \textsf{0.91}  &  \textsf{0.78}  \\
 12  &  \cellcolor[gray]{0.9}\textsf{37.91}  &  \textsf{1.08}  &  \textsf{1.12}  &  \textsf{1.03}  &  \textsf{0.93}  &  \textsf{1.03}  & \best{\textsf{1.13}}  &  \textsf{1.03}  &  \textsf{1.08}  \\
 16  &  \cellcolor[gray]{0.9}\textsf{32.41}  &  \textsf{1.11}  &  \textsf{1.04}  &  \textsf{1.06}  & \best{\textsf{1.13}}  &  \textsf{1.07}  &  \textsf{1.07}  &  \textsf{1.02}  &  \textsf{1.10}  \\
 20  &  \cellcolor[gray]{0.9}\textsf{28.63}  &  \textsf{1.05}  &  \textsf{1.09}  &  \textsf{1.24}  & \best{\textsf{1.35}}  &  \textsf{1.08}  &  \textsf{1.04}  &  \textsf{1.04}  &  \textsf{1.15}  \\
 24  &  \cellcolor[gray]{0.9}\textsf{27.25}  &  \textsf{1.18}  &  \textsf{1.39}  & \best{\textsf{1.59}}  &  \textsf{1.53}  &  \textsf{1.10}  &  \textsf{1.12}  &  \textsf{1.14}  &  \textsf{1.40}  \\
 28  &  \cellcolor[gray]{0.9}\textsf{24.91}  &  \textsf{1.20}  &  \textsf{1.51}  & \best{\textsf{1.67}}  &  \textsf{1.41}  &  \textsf{1.05}  &  \textsf{1.17}  &  \textsf{1.30}  &  \textsf{1.50}  \\
 32  &  \cellcolor[gray]{0.9}\textsf{23.63}  &  \textsf{1.39}  & \best{\textsf{1.63}}  &  \textsf{1.55}  &  \textsf{1.31}  &  \textsf{1.20}  &  \textsf{1.27}  &  \textsf{1.41}  &  \textsf{1.41}  \\
\hline
\hline
 8  &  \cellcolor[gray]{0.9}\textsf{61386.36}  &  \textsf{50.0}  &  \textsf{73.3}  &  \textsf{67.7}  &  \textsf{50.7}  &  \textsf{56.3}  &  \textsf{81.3}  &  \textsf{89.0}  & \best{\textsf{94.9}}  \\
 12  &  \cellcolor[gray]{0.9}\textsf{36298.84}  &  \textsf{59.3}  &  \textsf{76.3}  &  \textsf{80.6}  &  \textsf{82.1}  &  \textsf{62.4}  &  \textsf{81.8}  &  \textsf{89.3}  & \best{\textsf{93.2}}  \\
 16  &  \cellcolor[gray]{0.9}\textsf{19385.18}  &  \textsf{70.4}  &  \textsf{84.0}  &  \textsf{88.8}  &  \textsf{90.8}  &  \textsf{72.8}  &  \textsf{88.7}  &  \textsf{94.2}  & \best{\textsf{96.5}}  \\
 20  &  \cellcolor[gray]{0.9}\textsf{10325.29}  &  \textsf{74.6}  &  \textsf{88.3}  &  \textsf{93.7}  &  \textsf{96.1}  &  \textsf{78.8}  &  \textsf{92.9}  &  \textsf{97.0}  & \best{\textsf{98.5}}  \\
 24  &  \cellcolor[gray]{0.9}\textsf{6566.03}  &  \textsf{82.4}  &  \textsf{94.9}  &  \textsf{97.5}  &  \textsf{98.7}  &  \textsf{84.9}  &  \textsf{96.1}  &  \textsf{98.4}  & \best{\textsf{99.4}}  \\
 28  &  \cellcolor[gray]{0.9}\textsf{3141.06}  &  \textsf{82.8}  &  \textsf{94.4}  &  \textsf{98.0}  &  \textsf{99.1}  &  \textsf{85.2}  &  \textsf{96.2}  &  \textsf{98.8}  & \best{\textsf{99.5}}  \\
 32  &  \cellcolor[gray]{0.9}\textsf{2399.46}  &  \textsf{88.3}  &  \textsf{97.1}  &  \textsf{99.1}  &  \textsf{99.7}  &  \textsf{89.6}  &  \textsf{97.8}  &  \textsf{99.3}  & \best{\textsf{99.8}}  \\
 \hline
\hline
\end{tabular*}\\[0.2cm]
\end{scriptsize}
\caption{\label{periodic3}Experimental results on a \textsc{Period-40} short integer sequence.}
\end{table}

\subsection*{Experimental Results on Periodic Sequences}
Experimental results on \textsc{Period-$\delta$} problem have been conducted on a periodic sequence with a period equal to $10$ and with $\delta=5$ (see Table \ref{periodic1}). The results show as the \textsc{Nr}$1$ approach is the best choice in most of the  cases, achieving a speed up of $1.3$ in suitable conditions. However in some cases the \textsc{Fct} algorithm turns out to be the best choice especially on short patterns. 
The \textsc{No} approach is always less efficient of the \textsc{Fct} algorithm although the gain in number of detected false positives is always between $65\%$ and $95\%$. This behavior is due to the high number of candidate occurrences detected by the algorithm, despite its gain in number of false positives, and to the relative effort in the construction of the filters values.

When the size of $\delta$ increases (see  Table \ref{periodic2} and  Table \ref{periodic3} ) the performances of the \textsc{No} approach get better achieving a speed up of $1.4$ in the best cases. However the \nr approach turns out to be always the best solutions with a speed up close to $1.7$ for long patterns.

The gain in number of false positives is always in the range between $50\%$ and $99.7\%$ for the \textsc{Nr} algorithm, and between $80\%$ and $99.8\%$ in the case of the \textsc{No} algorithm. The gain of the \textsc{No}$4$ algorithm is in most cases close the $100\%$.

\section{Conclusions} \label{sec:conclusions}
In this paper we discussed the Order Preserving Pattern Matching Problem and presented two new families of filtering approaches to solve such problem which turn out to be much more effective in practice than the previously presented methods. The presented methods translate the original sequence on new sequences over large alphabets  in order to speed up the searching process and reduce the number of false positives. 
From our experimental results it turns out that our proposed solutions are up to 2 times faster than the previous solutions reducing the number of false positives up to 99\% under suitable conditions.

\bibliographystyle{plain}

\end{document}